\documentclass{article}
\usepackage{amssymb}
\usepackage{amsmath}
\usepackage{amsthm}
\newtheorem{theorem}{Theorem}
\newtheorem{corollary}{Corollary}
\newtheorem{lemma}{Lemma}
\theoremstyle{definition}
\newtheorem{condition}{Condition}
\newtheorem{remark}{Remark}

\title{Redundancy Estimates for Word-Based Encoding of Sequences Produced by
a Bernoulli Source\footnote{Translation from Russian original:
"Ocenki izbytochnosti pri poslovnom kodirovanii soobscheniy,
porojdaemyh bernullievskim istochnikom", Problemy Peredachi
Informacii (Problems of Information Transmission), 8 (2)
(1972)~21--32. Translated by Yuriy A. Reznik, {\tt
yreznik@ieee.org}.}} \author{G. L. Khodak}
\date{}

\begin{document}
\maketitle

\begin{abstract}
The efficiency of a code is estimated by its redundancy $R$, while the
complexity of a code is estimated by its average delay $\bar N$. In this work
we construct word-based codes, for which $R \lesssim \bar N^{-5/3}$.
Therefore, word-based codes can attain the same redundancy as block-codes
while being much less complex.

We also consider uniform on the output codes, the benefit of which is the
lack of a running synchronization error. For such codes $\bar N^{-1} \lesssim
R \lesssim \bar N^{-1}$, except for a case when all input symbols are
equiprobable, when $R \leqslant \bar N^{-2}$ for infinitely many $\bar N$.
\end{abstract}

\section{Introduction}
Consider a Bernoulli source sequentially producing symbols from an input
alphabet ${a_1,\ldots,a_m}$ ~$\left(2 \leqslant m < \infty\right)$ with
probabilities ${p_1,\ldots,p_m}$, $\sum_{i=1}^m p_i = 1$, $p_i > 0$
~$\left(i=1,\ldots,m\right)$. The entropy of the source $H=-\sum_{i=1}^m p_i
\log_2 p_i$. Assume that a message is an infinite-length sequence of symbols
from the input alphabet ${a_{i_k}}_{k=1}^\infty$. It is necessary to map such
a message to a sequence of symbols from an output alphabet
${b_1,\ldots,b_n}$~$\left(2 \leqslant n < \infty\right)$, which is its code.
Such a mapping can be established by using word-based codes. Select a finite
set of words $A_j$~$\left(j=1,2,\ldots\right)$ from the input alphabet, such
that any message can be uniquely represented by a sequence of such words
(indeed, this immediately implies that words $A_j$ are prefix free; i.e. no
word is a prefix of another). In turn, words $A_j$ are represented by words
$\phi(A_j)$ from the output alphabet. A word-based code for a given message
is constructed as follows:
\begin{equation*}
\left\{a_{i_k}\right\}_{k=1}^\infty = \left\{A_{j_r}\right\}_{r=1}^\infty
\rightarrow \left\{\phi\left(A_{j_r}\right)\right\}_{r=1}^\infty =
\left\{b_{i_s}\right\}_{s=1}^\infty.
\end{equation*}

In this paper, we only consider decipherable encodings, i.e. ones such that
$\phi\left(A_{i_1}\right)\phi\left(A_{i_2}\right) \ldots$
$\phi\left(A_{i_s}\right) =$
$\phi\left(A_{j_1}\right)\phi\left(A_{j_2}\right) \ldots$
$\phi\left(A_{i_t}\right)$ always implies that $s=t$ and
$\phi\left(A_{i_k}\right)=\phi\left(A_{j_k}\right)$, $k=1,\ldots,s$.
Constructed codes have, in fact, an even more strong property, namely that
different messages have different codes.

In the terminology of V.~I.~Levenstein~\cite{Levenstein}, word-based code is
specified by a coding system $\left\{{A,U,B,V}\right\}$, where $A$ is the
input alphabet, $B$ is the output alphabet, $U$ is the set of words $A_j$, $V
= \phi\left(A_j\right)$, and it is required that $U$ is strongly (prefix-)
free, and that any message begins with a word in $U$. The number of letters
in a word $A$ (i.e. its length) is denoted by $\left| A \right|$. The code is
called a block-code, or uniform on the input code, if all words $A_j$ have
the same length. The code is called uniform on the output code, if all words
$\phi\left(A_j\right)$ have the same length.

The probability of a word $A=a_{i_1} \ldots a_{i_k}$ in the input alphabet is
denoted by $p\,(A)$. For Bernoulli source $p\,(A) = p_{i_1} \ldots p_{i_k}$.

The complexity of a code is estimated by using its delays: average $\bar N =
\sum_{j} p\left({A_j}\right) \left|{A_j}\right|$, and maximum $N = \max_{j}
\left|{A_j}\right|$. For block codes $\left|{A_j}\right| = \bar N = N$
($j=1,\ldots,m^n$).

The efficiency of a code is estimated by using its redundancy: $R = $ $\bar
N^{-1} \sum_j p\left({A_j}\right) \left|\phi\left({A_j}\right)\right| - H
\log_2^{-1} n$. C.~Shannon has shown that $0 \leqslant R \leqslant N^{-1}$
\cite{Shannon}. From the paper of V.~M.~Sidelnikov~\cite{Sidelnikov} it
follows that for all word-based codes $R \geqslant 0$. The redundancy shows
how the average number of output letters per each input letter is greater
than the minimum necessary. Note, that both redundancy and average delay are
continuous functions of probabilities of symbols $p_1,\ldots,p_m$.

R.~E.~Krichevski~\cite{Krichevski} has shown that for optimal block-codes $R
\gtrsim N^{-1}$~($N \rightarrow \infty$), except for the sources with
coinciding fractional parts of $\log_n p_i$ \footnote{Here, as usual, the
notation $f \gtrsim g$ means that $\lim \frac{f}{g} > 0$. If $f > 0$, then
there exists a constant $c > 0$, such that for all arguments $f > cg$.
Assuming the existence of such an inequality, we, in some instances, may not
specify the direction of growth of the argument.}. In the present paper, we
construct word-based codes, for which $R \lesssim \bar N ^{5/3}$, $N \lesssim
\bar N \log \bar N$. Compared with block codes of the same redundancy our
codes are much less complex. It is proven, that for almost all Bernoulli
sources (we apply Lebesgue measure on points $\left( p_1,\ldots,p_{m-1}
\right)$) word-based codes satisfy: $R \gtrsim \bar N ^{-9} \log^{-8} \bar
N$.

Word based codes are susceptible to running synchronization errors, i.e. a
single error in an encoded message $\left\{b_{i_s}\right\}_{s=1}^{\infty}$,
may result in incorrect separation of words $\phi\left(A_j\right)$ in an
arbitrary large portion of the code, resulting in an arbitrary large number
of errors in the reconstructed message. Uniform on the output codes have an
advantage that the corresponding error in the reconstruction is limited to a
single word $A_j$. We construct uniform on the output codes, for which $R
\lesssim \bar N ^ {-1}$. It is proven, that if not all $p_i = 1/m$
($i=1,\ldots,m$), then $R \gtrsim \bar N ^ {-1}$. If $p_1 = \ldots = p_m =
1/m$, then for infinitely many positive integer $\bar N$: $R \lesssim \bar N
^ {-2}$.

\section{Relation between redundancy and lengths of
words~$\phi\left(A_j\right)$}

From the paper of V.~M.~Sidelnikov \cite{Sidelnikov} it follows that $H =
-\bar N^{-1} \sum_j p\left(A_j\right) \log_2 p\left(A_j\right)$. Using this
equation we arrive at:
\begin{equation}
R = \bar N^{-1} \sum_j p\left(A_j\right) \left({
\left|\phi\left(A_j\right)\right| + \log_n p\left(A_j\right) }\right).
\label{eq:1}
\end{equation}

We introduce the following notation:
\begin{eqnarray}
\delta & = & 1 - \sum_j n^{-\left|\phi\left(A_j\right)\right|}
\label{eq:2} \\
\varepsilon_j & = & \left|\phi\left(A_j\right)\right| +
\log_n p\left(A_j \right) ~~\left({j=1,2,\ldots}\right) \label{eq:3} \\
\varepsilon'_j & = & \left\{
\begin{array}{clclc}
-1 & \mbox{if} & \varepsilon_j & < & -1\,, \\
\varepsilon_j & \mbox{if} & \left|{\varepsilon_j}\right|
& \leqslant & 1\,, \\
1 & \mbox{if} & \varepsilon_j & > & 1\,.
\end{array}
\right.
\label{eq:4}
\end{eqnarray}

It is well known (see, e.g. \cite{McMillan}), that the necessary and
sufficient condition for the existence of a decipherable code with lengths of
codewords $\left|{\phi\left({A_j}\right)}\right|$
$\left({j=1,2,\ldots}\right)$ is given by the Kraft inequality $\delta
\geqslant 0$.
\begin{theorem}
The redundancy of a decipherable code satisfies:
\begin{equation*}
R \geqslant \bar N^{-1} \left({\delta \ln^{-1} \delta + \frac{1}{2\,n} \ln n
\sum_j p\left({A_j}\right){\varepsilon'_j} ^2 }\right)\,.
\end{equation*}
If ~$\left|{\varepsilon_j}\right| \leqslant 1$ for all $j=1,2,\ldots$, then
\begin{equation*}
R \leqslant \bar N^{-1} \left({\delta \ln^{-1} \delta + \frac{n}{2} \ln n
\sum_j p\left({A_j}\right)\varepsilon_j ^2 }\right)\,.
\end{equation*}
\end{theorem}

\begin{proof}
Decompose $n^{\varepsilon_j}$ in a Taylor series $\left(j=1,2,\ldots\right)$,
\begin{equation}
n^{\varepsilon_j} = 1 - \varepsilon_j \ln n + \eta\left(\varepsilon_j\right)
~. \label{eq:5}
\end{equation}
The remainder
\begin{equation}
\eta\left(\varepsilon_j\right) = n^{\varepsilon_j} - 1 +
\varepsilon_j \ln n ~. \label{eq:6}
\end{equation}
From the sign of $\frac{d}{d\,\varepsilon_j} \eta\left(\varepsilon_j\right)$
it follows, that
\begin{equation}
\eta\left(\varepsilon'_j\right) \leqslant \eta\left(\varepsilon_j\right)~.
\label{eq:7}
\end{equation}
Since $\varepsilon'_j \in \left[-1,1\right]$, the Lagrange estimate for the
remainder is
\begin{equation}
\frac{\ln^2 n}{2\,n} \, {\varepsilon'_j}^2 \leqslant
\eta\left(\varepsilon'_j\right) \leqslant \frac{n \ln^2 n}{2} \,
{\varepsilon'_j}^2~. \label{eq:8}
\end{equation}
By multiplying (\ref{eq:5}) by $p\left({A_j}\right)$ and summing all terms
over $j$, we obtain
\begin{equation}
\sum_j p\left({A_j}\right) n^{-\varepsilon_j} = \sum_j p\left({A_j}\right)
\left({1 - \varepsilon_j \ln n + \eta\left(\varepsilon_j\right) }\right)~.
\label{eq:9}
\end{equation}

From (\ref{eq:3}) we have
\begin{equation}
p\left({A_j}\right) n^{-\varepsilon_j} -
n^{-\left|{\phi\left({A_j}\right)}\right|}~, \label{eq:10}
\end{equation}
and from (\ref{eq:1}) and (\ref{eq:3})
\begin{equation}
\sum_j p\left({A_j}\right) \varepsilon_j = R \bar N~. \label{eq:11}
\end{equation}

From (\ref{eq:2}), (\ref{eq:9}-\ref{eq:11}) it follows, that
\begin{equation}
\bar N \, R = \ln^{-1} n \left({ \delta + \sum_j p\left({A_j}\right)
\eta\left(\varepsilon_j\right)}\right)~. \label{eq:12}
\end{equation}
The statement of the theorem follows from (\ref{eq:7}), (\ref{eq:8}), and
(\ref{eq:12}).
\end{proof}

By $\|x\|$ we denote the distance of real number $x$ to its nearest integer.
\begin{corollary}
The following inequality holds
\begin{equation*}
R \geqslant \bar N^{-1} \frac{\ln n}{2\,n} \sum_j p\left({A_j}\right)
\left\|{p\left({A_j}\right)}\right\|~.
\end{equation*}
\end{corollary}
This follows from the first claim of the Theorem 1, Kraft inequality, and an
observation that $\left|{\varepsilon'_j}\right| \geqslant
\left\|{p\left({A_j}\right)}\right\|$.

\section{On approximation of linear forms by integer numbers}

From Theorem 1 and the Corollary it follows that the redundancy (of a
word-based code) depends on quantities $\left\| p\left({A_j}\right)
\right\|$. If $k_i$ is a number of letters $a_i$ in a word $A$, then $\log_n
p(A) = \sum_{i=1}^m k_i \log_n p_i$ is a linear form of $k_i$.

Consider an arbitrary linear form $f\left(k_1,\ldots,k_m\right) =
\sum_{i=1}^m k_i d_i$, where coefficients $d_i$ are fixed, and $k_i$
~$\left(i=1, \ldots, m\right)$ are integer numbers.

By $[x]$ and $\{x\}$ we denote the integer and fractional parts of a real
number $x$ correspondingly; $\|x\| = \min \left({ \{x\}, 1-\{x\} }\right)$.
We will also need the following obvious relationships ($x$, $y$ are real
numbers, $l$ is an integer):
\begin{equation}
\{x+l\} = \{x\}~, \label{eq:13}
\end{equation}
\begin{equation}
\{x+y\} \leqslant \{x\} + \{y\}~, \label{eq:14}
\end{equation}
\begin{equation}
\mbox{if}~\{x\} \geqslant \{y\}, \mbox{then}~\{x-y\} = \{x\} - \{y\}~.
\label{eq:15}
\end{equation}

\begin{lemma}
If $d_m$ is irrational, then there exists infinitely many integers $T$, such
that for any vector $\left({k_1, \ldots, k_m}\right)$ there exist numbers
$k'_m$ and $k''_m$, $0 \leqslant k'_m < T$, $0 \leqslant k''_m < T$, such
that:
\begin{eqnarray*}
\left\{f\left({k_1, \ldots, k_m+k'_m}\right)\right\} & \leqslant &
2/T~, \\
1 - \left\{f\left({k_1, \ldots, k_m+k''_m}\right)\right\} & \leqslant & 2/T~.
\end{eqnarray*}
\end{lemma}
\begin{proof}
For $T$ we pick a denominator of any fraction giving the best approximation
to $d_m$, except for the first one~\cite[Chapter~1,~\S~2,~p.~2]{Cassels}.
Let $\widetilde{T}$ be a denominator of the preceding fraction. It has been
shown in \cite[Chapter~1,~\S~2,~p.~3]{Cassels}, that either one of the
following two statements holds
\begin{equation}
\left\{{\widetilde{T}\,d_m}\right\} \leqslant T^{-1} ~\mbox{and}~ 1 -
\left\{{T\,d_m}\right\} \leqslant T^{-1}~, \label{eq:16}
\end{equation}
\begin{equation}
1 - \left\{{\widetilde{T}\,d_m}\right\} \leqslant T^{-1} ~\mbox{and}~
\left\{{T\,d_m}\right\} \leqslant T^{-1}~. \label{eq:17}
\end{equation}
Our proof is the same in both cases. So, for simplicity, assume that the
correct statement is (\ref{eq:16}).

We prove the existence of $k'_m$ (the existence of $k''_m$ can be proven in
the same way). Take an arbitrary vector $\left({k_1, \ldots, k_m}\right)$.
Since $d_m$ is irrational, then there exists $k$ such that
\begin{equation}
\left\{f\left({k_1, \ldots, k_{m-1}, k}\right)\right\} \leqslant 2/T
\label{eq:18}
\end{equation}
(see \cite[Chapter~4,~\S~3]{Cassels}).
Let us prove that
\begin{equation}
\left\{f\left({k_1, \ldots, k_{m-1}, k + T}\right)\right\} \leqslant 2/T
~\mbox{or}~ \left\{f\left({k_1, \ldots, k_{m-1}, k +
\widetilde{T}}\right)\right\} \leqslant 2/T~, \label{eq:19}
\end{equation}
and also
\begin{equation}
\left\{f\left({k_1, \ldots, k_{m-1}, k - T}\right)\right\} \leqslant 2/T
~\mbox{or}~ \left\{ f\left({k_1, \ldots, k_{m-1}, k - \widetilde{T}} \right)
\right\} \leqslant 2/T~. \label{eq:20}
\end{equation}
If
\begin{equation}
\left\{f\left({k_1, \ldots, k_{m-1}, k}\right)\right\} \leqslant 1/T~,
\label{eq:21}
\end{equation}
then, from (\ref{eq:16}), (\ref{eq:14}), and (\ref{eq:15}) it follows that
\begin{eqnarray}
\lefteqn{\left\{f\left({k_1, \ldots, k_{m-1}, k + \widetilde{T}}\right)
\right\} = \left\{f\left({k_1, \ldots, k_{m-1}, k}\right) + \widetilde{T} \,
d_m \right\}} \nonumber \\
& & \leqslant \left\{f\left({k_1, \ldots, k_{m-1}, k}\right)\right\} +
\left\{\widetilde{T} \, d_m \right\} \leqslant 2/T~. \label{eq:22}
\end{eqnarray}
At the same time, if (\ref{eq:21}) is false, then from (\ref{eq:18}) we have
\begin{equation*}
1/T < \left\{f\left({k_1, \ldots, k_{m-1}, k}\right)\right\} \leqslant 2/T~.
\end{equation*}
In this case, from (\ref{eq:16}), (\ref{eq:13}), and (\ref{eq:15}) it follows
that
\begin{equation}
\left\{ f\left( k_1, \ldots, k_{m-1}, k + T \right) \right\} = \left\{
f\left( {k_1, \ldots, k_{m-1}, k}\right) - \left(1 - \left\{T \, d_m\right\}
\right)\right\} \leqslant 2/T~. \label{eq:23}
\end{equation}
From (\ref{eq:22}) and (\ref{eq:23}) follows (\ref{eq:19}). Statement
(\ref{eq:20}) can be proven in the same way.

Based on (\ref{eq:19}) and (\ref{eq:20}) it is clear that for every $k$
satisfying condition (\ref{eq:18}) there exist smaller and greater numbers at
the distance not exceeding $T$ (and not lesser than 1) that also satisfy
condition (\ref{eq:18}). Therefore, $k_m$ lies within some pair of such
numbers, with distance (between these numbers) not larger than $T$, which
proves the lemma.
\end{proof}

\section{Estimates of the average and maximal lengths of words in some sets}

Hereafter, unless the contrary is stated, we assume that words are taken from
an input alphabet $\left\{a_1, \ldots, a_m\right\}$. In this section, we
obtain an estimate for the average length and cumulative probability of words
of sufficiently large lengths for a given selection of words in a set,
conforming, in particular, conditions of Lemma 1. Proofs of these estimates
are omitted, but they can be easily reconstructed by using the statements and
the order of lemmas in this section.

By $k(A)$ we denote a vector $\left(k_1, \ldots, k_m\right)$, where each
coordinate $k_i$ is the number of letters $a_i$ in a word $A$. We call such a
vector $k(A)$ a profile of the word $A$.
Let also $t(A) = \sum_{i=1}^{m-1} k_i$. By definition of word length $|A| =
\sum_{i=1}^m k_i$, and by definition of probability $p\,(A) = p_1^{k_1}
\ldots p_m^{k_m}$.

By $A'A''$ we denote a result of catenation of words $A'$ and $A''$. In
accordance with definitions:
\begin{eqnarray*}
k\left(A'A''\right) & = & k\left(A'\right) + k\left(A''\right),  \\
t\left(A'A''\right) & = & t\left(A'\right) + t\left(A''\right),  \\
\left|A'A''\right|  & = & \left|A'\right| + \left|A''\right|,    \\
p\left(A'A''\right) & = & p\left(A'\right) \, p\left(A''\right).  \\
\end{eqnarray*}

Assume that a set of all words contains also an empty word, $\lambda$. For
such a word: $k(\lambda) = (0,\ldots,0)$, $p(\lambda) = 1$, and for any words
$A$: $\lambda \, A = A \, \lambda = A$.

In what follows, all numbers, except for probabilities of symbols, and
constants in estimates of $\left(c_1, \ldots, c_m\right)$, are assumed to be
non-negative integers.

Each set $\mathfrak M$ of vectors $\left(k_1, \ldots, k_m\right)$ can be
associated with a set of words $M$. Suppose that $A \in M$ if and only if
$k(A) \in \mathfrak M$, and $A$ cannot be decomposed into $A'A''$, such that
$k\left(A'\right) \in \mathfrak M$, and $A'' \neq \lambda$. I.e. $M$ is a
prefix-free set.
\begin{lemma}
Given any set $\mathfrak M$ and a word $A$, if $k(A) \in \mathfrak M$, then
$A$ can be presented as $A'A''$, where $A' \in M$.
\end{lemma}
\begin{condition}
We say that a set $\mathfrak M$ of vectors $\left(k_1, \ldots, k_m\right)$
satisfies Condition 1 with parameter $T$, if for each $s \geqslant 1$ and
each vector $\left(k_1, \ldots, k_m\right)$, such that $\sum_{i=1}^{m-1}k_i =
sT^2$, there exists $k'_m$, such that $0 \leqslant k'_m < T$ and $\left(k_1,
\ldots, k_{m-1}, k'_m + k_m\right) \in \mathfrak M$.
\end{condition}

By $F(D)$ we denote a set of words $A = a_{i_1} \ldots a_{i_r}$, such that
$a_i \neq a_m$, and $t(A) = D$. Let also $F(0) = \lambda$. It is clear that
$F(D)$ is a prefix-free set.
\begin{lemma}
For any $D_i \geqslant 1$, such that $\sum_i D_i = D$, any word $A \in F(D)$
has a unique decomposition into $A_1 A_2 \ldots A_i \ldots$, where $A_i \in
F(D_i)$~$(i=1,2,\ldots)$.
\end{lemma}
\begin{lemma}
Let $D \geqslant 1$ and $A = a_{i_1} \ldots a_{i_r} \in F(D)$. If $a_{i_1} =
a_m$, then $a_{i_2} \ldots a_{i_r} \in F(D)$ and vice verse. If $a_{i_1} \neq
a_m$, then $a_{i_2} \ldots a_{i_r} \in F(D-1)$ and vice verse.
\end{lemma}
\begin{lemma}
If $M$ is a prefix-free set, then for any word $A'$
\begin{equation*}
\sum_{A: A'A \in M} p\,(A) \leqslant 1.
\end{equation*}
\end{lemma}
\begin{lemma}
For all $D \geqslant 1$
\begin{equation*}
\sum_{A \in F(D)} p\,(A) = 1.
\end{equation*}
\end{lemma}
\begin{lemma}
There exists a constant $c_1 > 0$, such that for each $\mathfrak M$, that
satisfies Condition 1 with parameter $T$, any $s$, and any word $A' \in
F\left(s \, T^2 \right)$, holds
\begin{eqnarray*}
\sum_{A \in F(T^2)} p\,(A) & \geqslant & c_1 T^{-1} \\
k\left( A' A \right) & \in & \mathfrak M.
\end{eqnarray*}
\end{lemma}

By $F_1(D,M)$ denote a set of words $A \in F(D)$, which cannot be decomposed
into $A'A''$, where $A' \in M$, and $A'' \neq \lambda$.
\begin{lemma}
For any $\mathfrak M$, satisfying Condition 1 with parameter $T$, and any $s
\geqslant 1$, the following holds
\begin{equation*}
\sum_{A \in F_1\left(s T^2, M\right)} p\,(A)\, \leqslant \left(1 - c_1
\,T^{-1} \right)^s~,
\end{equation*}
where $c_1$ is a constant, existence of which is guaranteed by Lemma 7.
\end{lemma}
\begin{lemma}
The following holds:
\begin{equation*}
\sum_{A \in F(D)} p\,(A)\, |A| = \frac{D}{1-p_m}~.
\end{equation*}
\end{lemma}
The main result in this section is given by the following lemma.
\begin{lemma}
For any $\mathfrak M$, satisfying Condition 1 with parameter $T$, the
following holds:
\begin{equation*}
\sum_{A \in M} p\,(A)\, |A| \lesssim T^3~,~~(T \rightarrow \infty)~.
\end{equation*}
\end{lemma}

Now, given a fixed number $T$, we would like to find out how to select the
minimum length $T_2$ of words, such that their combined probability is
sufficiently small. Such a result will be needed for estimating the maximum
delay of the code.
\begin{lemma}
There exists $T_2 = T_2(T)$, such that
\begin{equation*}
T_2 \lesssim T^3 \ln T~,~~(T \rightarrow \infty)~,
\end{equation*}
and for any $\mathfrak M$, satisfying Condition 1 with parameter $T$, the
following holds
\begin{equation*}
\sum_{A \in M,~|A| \geqslant T^2} p\,(A) \lesssim T^{-2}~,~~(T \rightarrow
\infty)~.
\end{equation*}
\end{lemma}

\section{Construction of the code}

As we pointed out in Section~1, in order to construct a (word-based) code one
needs to specify a set of words $A_j$~$(j=1,2,\ldots)$, such that any
incoming message can be uniquely represented by them. In addition, words
$A_j$ need to be mapped to output words $\phi \left( A_j \right)$, such that
the resulting code is decipherable. Hereafter, we assume that all words are
not empty.

Let $M'$ and $M''$ be some sets of words. By $M' \wedge M''$ we denote a
prefix-free extension of $M'$ by words from $M''$. In other words, $M' \wedge
M''$ is a set of words from $M' \cup M''$, which cannot be presented as
$A'A''$, where $A' \in M' \cup M''$, and $A''$ is not empty. It is clear,
that $M' \wedge M''$ is prefix-free and that the operation $\wedge$ is
commutative and associative. If $M$ is prefix-free, then $M \wedge M = M$.
\begin{lemma}
If any message begins with a word from $M'$, then it can also be uniquely
represented by words from $M' \wedge M''$, with any extension set $M''$.
\end{lemma}

The proof follows from the definition of the operation $\wedge$.

\begin{theorem}
For any Bernoulli source and infinitely many $T$ there exist decipherable
codes such that
\begin{equation*}
R \lesssim \bar{N}^{-1} \, T^{-2}, ~~ \bar{N} \lesssim T^3, ~~ N \lesssim T^3
\ln T ~~ (T \rightarrow \infty)~.
\end{equation*}
\end{theorem}

\begin{proof}
a) First, consider a case when not all $\log_n p_i$~ $(i=1,\ldots,m)$ are
rational. With no loss of generality, we can assume that the last such a
number $\log_n p_m$ is irrational.

Let $T$ be one of the numbers satisfying conditions of Lemma~1 for a linear
form $-\sum_{i=1}^m k_i \log_n p_i$, and $T_2 = T_2(T)$ a number, satisfying
conditions of Lemma~11.

Consider a set $\widetilde{\mathfrak M}_1 \left(\widetilde{\mathfrak M}_2
\right)$ of vectors $\left(k_1,\ldots,k_m\right)$, such that
\begin{equation*}
\left\{-\sum_{i=1}^m k_i \log_n p_i\right\} \leqslant \frac{2}{T} \left({1 -
\left\{-\sum_{i=1}^m k_i \log_n p_i\right\} }\right) \leqslant \frac{2}{T}~.
\end{equation*}
According to Lemma~1, the sets $\widetilde{\mathfrak M}_1$ and
$\widetilde{\mathfrak M}_2$ are not empty, and satisfy the Condition 1 with
parameter $T$. Let:
\begin{equation*}
{\mathfrak M}_i = \widetilde{\mathfrak M}_i \cup \left\{ \left(k_1, \ldots,
k_m \right) \left|~ {\sum_{i=1}^m k_i = T_2} \right. \right\}~.
\end{equation*}
The sets ${\mathfrak M}_1$ and ${\mathfrak M}_2$ also satisfy the Condition~1
with parameter $T$. Let $M_1$ and $M_2$ be the sets of words that are
associated with the sets of vectors ${\mathfrak M}_1$ and ${\mathfrak M}_2$
correspondingly (see Section~4 for details). Let $\left\{a_{i_k}
\right\}_{k=1}^\infty$ be some message. Then $k\left(a_{i_1} \ldots
a_{i_{T_2}}\right) \in {\mathfrak M}_i$, and according to Lemma 2, such a
message begins with some word in $M_i$~$(i=1,2)$. Therefore, for any $A \in
M_1 \cup M_2$
\begin{equation}
\left| A \right| \leqslant T_2~. \label{eq:24}
\end{equation}

According to Lemma~5
\begin{equation}
\sum_{A \in M_i} p \left( A \right) \left| A \right| \lesssim T^3~~ (i=1,2)~.
\label{eq:25}
\end{equation}
From Lemma~11 and (\ref{eq:24})
\begin{equation}
\sum_{A \in M_i,~|A| = T^2} p \left( A \right) \lesssim T^{-2}~,~~(i=1,2)~.
\label{eq:26}
\end{equation}
Let us now define
\begin{equation}
l \left(A \right) = \left \{ {
\begin{array}{lcl}
\left[{-\log_n p \left( A \right)}\right]~, & \mbox{if} & A \in M_1,
~A \notin M_2 \,, \\
\left[{-\log_n p \left( A \right)}\right] + 1~, & \mbox{if} & A \in M_2\,.
\end{array}
} \right. \label{eq:27}
\end{equation}
If
\begin{equation}
\sum_{A \in M_1} n^{-\, l \left( A \right)} \leqslant 1~, \label{eq:28}
\end{equation}
then words $A_j$ can be taken from $M_1$. Lemma 12 ensures that $M_1 = M_1
\wedge M_1$ has the required properties.

Let
\begin{equation}
\sum_{A \in M_1} n^{- \, l \left( A \right)} > 1~. \label{eq:29}
\end{equation}
From (\ref{eq:27}) it follows, that
\begin{equation}
\sum_{A \in M_2} n^{- \, l \left( A \right)} \leqslant 1~. \label{eq:30}
\end{equation}
We will assume that
\begin{equation}
\sum_{A \in M_1 \wedge M_2} n^{- \, l \left( A \right)} \leqslant 1~.
\label{eq:31}
\end{equation}
In the contrary is true, we can simply exchange positions of $M_1$ and $M_2$
in the following construction procedure. Let us enumerate words in $M_2$,
$M_2 = \left\{ A^s, s = 1,2,\ldots \right\}$.  Consider
\begin{equation*}
g(k) = \sum_{A \in M_1 \wedge \left( \bigcup_{s=1}^k A^s \right)} n^{- \, l
\left( A \right)}~.
\end{equation*}
Due to (\ref{eq:29}) and (\ref{eq:31}) there exists $k_0$, such that
\begin{equation}
g\left( k_0 - 1 \right) > 1 \geqslant g \left( k_0 \right)~.
\label{eq:32}
\end{equation}
Based on (\ref{eq:27}) for any $k$
\begin{equation}
g\left( k - 1 \right) - g \left( k \right) \leqslant n \, p \left( A^k
\right)~. \label{eq:33}
\end{equation}
Since $\left| A^k \right| \geqslant T$ for any $k$, then
\begin{equation}
p \left( A^k \right) \leqslant \left( \max_{1 \leqslant 1 \leqslant m} p_i
\right)^T ~. \label{eq:34}
\end{equation}

We will take words $A_j$ from $M_1 \wedge \left( \bigcup_{s=1}^{k_0} A^s
\right)$. The uniqueness of the representation is guaranteed by Lemma 12.

If (\ref{eq:28}) holds, then from (\ref{eq:27}) it follows, that for $T > 4$
\begin{equation}
0 \leqslant 1 - \sum_j n^{- \, l \left( A_j \right)} \leqslant \sum_{j:\, A_j
\in M_2} p \left( A_j \right) = \sum_{j:\, \left| A_j \right| = T_2} p \left(
A_j \right)~. \label{eq:35}
\end{equation}
Using (\ref{eq:26}) and (\ref{eq:35}) we obtain
\begin{equation}
0 \leqslant 1 - \sum_j n^{- \, l \left( A_j \right)} \lesssim T^{-2}~.
\label{eq:36}
\end{equation}
If (\ref{eq:29}) holds, then using (\ref{eq:32}) we also arrive at
(\ref{eq:36}).

Observe that (\ref{eq:36}) is a Kraft inequality for a coding system with
code lengths $\left\{ l \left( A_j \right) \right\}$. This means, that there
exists a decipherable prefix code with $\left| \phi \left( A_j \right)
\right| = l \left( A_j \right)$~ $(i=1,2,\ldots)$ (see \cite{Huffman}). The
redundancy of such a code provides an upper bound for the redundancy of the
optimal one, which can be found by using Huffman technique \cite{Huffman}.

From (\ref{eq:27}) it follows that for any $j$~ $\left| \varepsilon_j \right|
\leqslant 1$ (see \cite[\S 2]{Shannon}). If $\left| A_j \right| < T_2$, then
$k\left( A_j \right) \in \widetilde{\mathfrak M}_1 \cup \widetilde{\mathfrak
M}_2$, and therefore, due to (\ref{eq:27})
\begin{equation}
\left| \varepsilon_j \right| \leqslant \frac{2}{T}~. \label{eq:37}
\end{equation}

From (\ref{eq:26}) we have
\begin{equation}
\sum_{j:\, \left| A_j \right| = T_2} p \left( A_j \right) \leqslant
\sum_{A \in M_1:\, \left| A \right| = T_2} p \left( A_j \right) +
\sum_{A \in M_2:\, \left| A \right| = T_2} p \left( A_j \right) \lesssim
T^{-2}~. \label{eq:38}
\end{equation}
From (\ref{eq:37}) and (\ref{eq:38}) we obtain
\begin{equation}
\sum_{j} p \left( A_j \right) \, \varepsilon_j^2 \lesssim T^{-2}~.
\label{eq:39}
\end{equation}

From (\ref{eq:36}), (ref{eq:39}), and the second claim of the Theorem~1, it
follows that
\begin{equation}
R \lesssim \bar{N}^{-1} T^{-2}~, \label{eq:40}
\end{equation}
while from (\ref{eq:25}) it follows that
\begin{equation}
\bar{N} \leqslant \sum_{A \in M_1} p\left(A\right) \left|A\right| + \sum_{A
\in M_2} p\left(A\right) \left|A\right| \lesssim T^3~. \label{eq:41}
\end{equation}
According to Lemma~11
\begin{equation}
N \leqslant T_2 \lesssim T^3 \ln T~. \label{eq:42}
\end{equation}
This completes the proof of the Theorem for the irrational $\log_n p_m$ case.

b) All $\log_n p_i$~ $(i=1,\ldots,m)$ are rational. We use the same
techniques and ideas as in the previous case. However, here it is possible to
prove an even stronger statement, namely that the redundancy can be made
arbitrary small using a constrained average delay, and that it decays
exponentially with the growth of the maximum delay.
\end{proof}

\begin{corollary}
The estimate $R \lesssim \bar{N}^{5/3}$ holds. This follows immediately from
the first two inequalities in the proof of Theorem 2.
\end{corollary}

\begin{corollary}
For infinitely many $T$ there exist codes such that
\begin{equation*}
R \lesssim \bar{N}^{5/3}~, ~~ N \lesssim \bar{N} \ln \bar{N}~.
\end{equation*}
\end{corollary}

\begin{proof}
Consider a case a) first. In order to construct a code we select words $A$
such that their~$t(A)$ are multiple of $T^2$, and vectors of compositions of
different words, say $\left( k_1, \ldots, k_m \right)$ and $\left( k'_1,
\ldots, k'_m \right)$, are either the same, or $\left| k_m - k'_m \right|
\geqslant 1 / 3 T$ (but the Condition~1 still holds). Then all claims of
Theorem 2 remain correct, but, at the same time $\bar{N} \gtrsim T^3$. This
fact, combined with (\ref{eq:42}) leads to an expression claimed by this
Corollary. The proof of the case b) is obtained in essentially the same way.
\end{proof}

%

In conclusion, we provide a very simple example of construction of such a
code. We deal with an input alphabet $\{a,b\}$, probabilities $p\,(a) = 0.4$,
$p\,(b) = 0.6$, entropy $H = 0.971$, and output alphabet ${0,1}$. We have a
case a). The corresponding linear form $f\left( k_1, k_2 \right) = 1.322\,
k_1 + 0.737\, k_2$. For simplicity, instead of searching for the denominators
of all suitable fractions, we will directly specify the accuracy of the
approximation of $f\left( k_1, k_2 \right)$ by integer numbers (the accuracy
used for code construction in Theorem 2 is $2/T$).
\begin{equation*}
\begin{array}{llcccccccc}
M_1  & a~~ & baa & bab & bba & bbb & ~   & ~   & ~   & ~   \\
M_2  &     &     &     & bba & bbb & ab~ & ba~ & aaa & aab \\
l(A) & 1   & 3   & 2   & 3   & 3   & 3   & 3   & 4   & 4
\end{array}
\end{equation*}
Let the accuracy be $0.3$. In $\widetilde{\mathfrak M}_1$ we include all
non-zero vectors $\left(k_1,k_2\right)$, such that $\left\{f\left( k_1, k_2
\right) \right\} \leqslant 0.3$, while in $\widetilde{\mathfrak M}_2$ we
include vectors, such that $\left\{f\left( k_1, k_2 \right) \right\}
\geqslant 0.7$. Thus $(1,0) \in \widetilde{\mathfrak M}_1$, $(1,1) \in
\widetilde{\mathfrak M}_2$, while vectors $(0,1)$, $(2,0)$, $(0,2)$ belong
to neither of these sets.
Let $T_2 = 3$. Now we can find $M_1$, $M_2$, and $l(A)$.
%
We obtain
\begin{equation*}
\sum_{A \in M_1} 2^{-l(A)} = \frac{9}{8} > 1~, ~~ \sum_{A \in M_2} 2^{-l(A)}
= \frac{5}{8} < 1~.
\end{equation*}
We also have $M_1 \wedge M_2 = \{a, ba, bba, bbb\}$, $\sum_{A \in M_1 \wedge
M_2} 2^{-l(A)} = \frac{7}{8} < 1$. Therefore $M_1$ has to be sequentially
combined with words from $M_2$, but the only non-trivial extension is a word
$ba$, since the other words in $M_2$ are either present in $M_1$ already, or
are extensions of the word $a \in M_1$. So, in our case $M_1 \wedge \{ba\} =
M_1 \wedge M_2$. We use $M_1 \wedge M_2$ as words $A_j$. Codewords
$\phi\left(A_j\right)$ can be found using Huffman technique:
\begin{equation*}
a \rightarrow 0\,, ~~ ba \rightarrow 10\,, ~~ bba \rightarrow 110\,, ~~ bbb
\rightarrow 111~.
\end{equation*}
For this code $\bar{N} = 1.96$, $N = 3$, $R = 0.029$.

%

\section{Construction of a uniform on the output code}

As it was pointed out in Section 1, the main advantage of the uniform on the
output codes is the lack of the running synchronization error.
\begin{theorem}
a) For any Bernoulli source and any $L \geqslant \log_n m$ there exists a a
decipherable code, such that
\begin{equation*}
\left|\phi \left( A_j \right) \right| = L ~~ \left(j=1,2,\ldots\right) ~~
\mbox{and} ~~ R \lesssim \bar{N}^{-1} ~~ \left( \bar{N} \rightarrow \infty
\right)~.
\end{equation*}

b) If $p_1=\ldots=p_m=1/m$, then $R \lesssim \bar{N}^{-2}$ for infinitely
many $L$.
\end{theorem}

\begin{proof}
a) Without any loss of generality we can assume that $p_m = \min_{1 \leqslant
i \leqslant m} p_i$, and therefore
\begin{equation}
- \log_n p_i \leqslant - \log_n p_m ~~~ (i = 1, \ldots, m-1)~. \label{eq:43}
\end{equation}
The number of blocks of length $L$ in the output alphabet is $n^L$. When $L
\geqslant \log_n m$ it will exceed the number of symbols in the input
alphabet. Therefore input symbols $a_i$~ $(i=1,\ldots,m)$ can be mapped to
different words $\phi\left( a_i \right)$ of length $L$, which result in a
decipherable code. In what follows, we construct a code for
\begin{equation}
L \geqslant - \log_n p_m ~. \label{eq:44}
\end{equation}
Consider a set $\mathfrak M$ of vectors
\begin{equation}
\left({k_1,\ldots,k_{m-1}, \left[{- \log_n^{-1} p_m \, \left({ L +
\sum_{i=1}^{L-1} k_i \log_n p_i }\right) }\right] }\right) ~, \label{eq:45}
\end{equation}
where $k_i = 0, 1, \ldots$~ $(i=1,\ldots,m)$, and $- \sum_{i=1}^{L-1} k_i
\log_n p_i \leqslant L$. Let $M$ be a set of words associated with $\mathfrak
M$ (see Section 4). Consider an arbitrary message $\left\{ a_{i_k}
\right\}_{k=1}^\infty$. Let $k\left( a_{i_1} \ldots a_{i_r} \right) = \left(
k_1(r), \ldots, k_m(r) \right)$~ $(r = 1, 2, \ldots)$, and $- \sum_{i=1}^m
k_i(r) \log_n p_i = h(r)$. From (\ref{eq:43}) and (\ref{eq:44}) it follows
that $h(1) \leqslant - \log_n p_m \leqslant L$, and for $r \rightarrow \infty
$, $h(r) \rightarrow \infty$: $h(r+1) \leqslant h(r) - \log_n p_m$.
Therefore, there exists a maximum number $r$, such that $h(r) \leqslant L$.
For such a number $r$
\begin{equation}
L + \log_n p_m < h(r) \leqslant L ~. \label{eq:46}
\end{equation}
From (\ref{eq:46}) we obtain
\begin{equation}
\frac{-1}{\log_n p_m} \, \left({L + \sum_{i=1}^m k_i(r) p_i }\right) - 1 <
k_m(r) \leqslant \frac{-1}{\log_n p_m} \, \left({L + \sum_{i=1}^m k_i(r) p_i
}\right)~. \label{eq:47}
\end{equation}
From (\ref{eq:47}) it follows that $k\left( a_{i_1} \ldots a_{i_r} \right)
\in {\mathfrak M}$. Therefore, according to Lemma 2, the message $\left\{
a_{i_k} \right\}_{k=1}^\infty$ begins with a word from $M$. So any message
begins with some word in $M$. Since $M$ is prefix-free, $M = M \wedge M$, and
from Lemma 12, it follows that any message can be uniquely represented by
words from $M$. Therefore, we can select $\left\{ A_j \right\} = M$.

Due to (\ref{eq:46})
\begin{equation}
L + \log_n p_m \leqslant - \log_n p\left( A_j \right) \leqslant L ~~ (j = 1,
2, \ldots)~. \label{eq:48}
\end{equation}
From (\ref{eq:48}) it follows that $p\left( A_j \right) \geqslant n^{-L}$,
and therefore, the number of words $A_j$ does not exceed $n^L$. Different
words $A_j$ can be mapped to different codes $\phi \left( A_J \right)$ of
length $L$, which results in a uniform on the output code. By using estimate
(\ref{eq:48}) in (\ref{eq:1}) (see Section 2), we arrive at
\begin{equation*}
R \leqslant \bar{N}^{-1} \sum_j p \left( A_j \right) \left( - \log_n p_m
\right) \lesssim \bar{N}^{-1}~,
\end{equation*}
which proves the first part of the theorem.

b) Let now $p_1=\ldots=p_m=1/m$. There exist infinitely many natural numbers
$X$ and $L$, such that
\begin{equation}
L - \frac{1}{X} \leqslant X \log_n m \leqslant L \label{eq:49}
\end{equation}
(see \cite[p.~3]{Cassels}).
As words $A_j$ we can select all possible combinations of input symbols of
length $X$. They all have probability $- X \log_n m$, and based on
(\ref{eq:49}) their number does not exceed $n^L$. Therefore, there exists a
decipherable code with $\left| \phi \left( A_j \right) \right| = L$.
Due to (\ref{eq:49}), the redundancy of such a code
\begin{equation*}
R \leqslant \bar{N}^{-1} \sum_j p \left( A_j \right) \frac{1}{X} =
\bar{N}^{-2}~,
\end{equation*}
since $X = \bar{N} = N$. This completes the proof.
\end{proof}
\begin{remark}
It is clear that $\bar{N} \gtrsim N \gtrsim \bar{N}$, $\bar{N} \gtrsim L
\gtrsim \bar{N}$.
\end{remark}

\section{Lower bounds for redundancy}

In the previous sections we have obtained the upper bounds for the
redundancy. In conclusion we will provide (without proofs) the lower bounds.

Bernoulli source is fully described by its probabilities $p_1, \ldots,
p_{m-1}$. If we use an $m-1$-dimensional Lebesgue measure for a set of points
$\left(p_1, \ldots, p_{m-1}\right)$, then the following holds.
\begin{theorem}
For almost all Bernoulli sources
\begin{equation*}
R \gtrsim \bar{N}^{-9} \ln^{-8} \bar{N} ~~ \left( \bar{N} \rightarrow \infty
\right) ~.
\end{equation*}
\end{theorem}

We give a sketch of the proof.

First we establish that vectors $\left(k_1, \ldots, k_m\right)$ for which
$\| - \sum_{i=1}^m k_i \log_n p_i \|$ is small are sufficiently isolated for
almost all sources. Then we obtain an estimate, similar (but inverse) to the
claim of Lemma 10. Finally we apply corollary of Theorem 1.

\begin{theorem}
If for some $i_0$: $p_{i_0} \neq 1/m$, then for uniform on the output code $R
\gtrsim \bar{N}^{-1}$ ~$\left(\bar{N} \rightarrow \infty\right)$.
\end{theorem}

We give a sketch of the proof.

First we find constants $c_6 > 0$, $c_7 > 0$, such that words with $\left| L
+ \log_n p \left( A_j \right) \right| \leqslant c_6$ ~ $\left( L = \left|
\phi \left( A_j \right) \right| \right)$ have a combined probability not
exceeding $c_7$. Then we apply the first inequality from Theorem 1.

\end{document}